\documentclass[12pt,reqno]{amsart}

\usepackage{graphicx}
\usepackage{amsthm,amsmath,amsfonts,amssymb,amsxtra,appendix,bookmark,dsfont,bm,color, braket}
\usepackage{cite}
\usepackage{mathtools}

\usepackage[truedimen,top=3truecm,bottom=3truecm,left=2.5truecm,right=2.5truecm]{geometry}

\newtheorem{theorem}{Theorem}[section]
\newtheorem{lemma}[theorem]{Lemma}

\theoremstyle{definition}

\newtheorem{remark}[theorem]{Remark}

\newtheorem{corollary}[theorem]{Corollary}

\renewcommand{\epsilon}{\varepsilon}
\newcommand{\E}{\mathbb{E}}

\renewcommand{\phi}{\varphi}
\newcommand{\R}{\mathbb{R}}

	\numberwithin{equation}{section}

\begin{document}

		\title[Asymptotics of Ground States in a Two-Cluster Region]{Asymptotics of Ground States for Helium-Like and Bosonic Atoms in a Two-Cluster Region}
		\author[Y. Goto]{Yukimi Goto}
		\address{Graduate School of Mathematical Sciences, University of Tokyo, Komaba, Meguro-ku, Tokyo 153-8914,
			Japan}
			\email{yukimi@ms.u-tokyo.ac.jp}

		\begin{abstract}
		We study the asymptotic behavior of positive ground states of $N$‑particle atomic Coulomb Hamiltonians in a two‑cluster region.
		For the $N$-particle ground state $\psi$, its one-particle density $\rho$, and $(N-1)$-particle ground state $\phi$, we show that for $0<\alpha<1$, as $|x_N|\to\infty$, $\psi(x_1,\dots,x_N)/\sqrt{\rho(x_N)}= \phi(x_1,\dots,x_{N-1})+O(|x_N|^{-2+\alpha} \ln |x_N|)$ uniformly in the region $|x_i|\le |x_N|^{\alpha}$.
		In the region $|x_i|\le R_0\ln |x_N|$, the convergence rate is $O(|x_N|^{-2} (\ln |x_N|)^2)$.
		We also establish an $O(|x_N|^{-2})$ bound on the squared $L^2$ error.
		These results hold both below the essential spectrum and at threshold.
		\end{abstract}
		\maketitle

		\section{Introduction}

		We consider the $N$-particle atomic Hamiltonian acting on $L^2(\R^{3N})$ defined by
		\[
		H_N=-\sum_{i=1}^N\frac{1}{2}\Delta_i+V_N(X_N),\quad
		V_N(X_N):=-\sum_{i=1}^N\frac{Z}{|x_i|}+\sum_{1\le i<j\le N}|x_i-x_j|^{-1}
		\]
		where $Z>0$ is the nuclear charge, and $\Delta_i$ denotes the three-dimensional Laplacian with respect to $x_i \in \R^3$ ,and $X_N=(x_1,\dots,x_N)$.
		Let $\psi \in L^2(\mathbb{R}^{3N})$ be the positive ground state satisfying $H_N\psi=E_N\psi$, where $E_N$ is the ground state energy of $H_N$ with $E_N \le E_{N-1}$:
		$$E_N=\inf\mathrm{spec} (H_N)=\inf\{\langle f,H_Nf\rangle \colon f\in H^2(\R^{3N}),\; \langle f,f \rangle =1\}.$$
		We use the convention $H_0=0$ and $E_0=0$.
		In addition, $\rho$ denotes the one-particle (marginal) density of $\psi$ defined by
		$$\rho(x_N):=\int_{\R^{3N-3}}|\psi(x_1,\dots,x_N)|^2\,dX_{N-1},\quad dX_{N-1}:=dx_1\cdots dx_{N-1}.$$

		We also consider the positive ground state $\phi$ corresponding to $E_{N-1}$, and define
		\begin{equation}
			\label{def.u}
			\begin{split}
		u(x_N)&:=\int_{\R^{3N-3}} \phi(x_1,\dots,x_{N-1})\psi (x_1,\dots, x_N)\,\,dX_{N-1}.
		\end{split}
		\end{equation}
				Throughout this paper, we assume that ground states are normalized: $\| \psi\|_{L^2(\R^{3N})}=1=\| \phi\|_{L^2(\R^{3N-3})}$.
				It should be noted that we do not impose the Pauli principle (fermionic anti-symmetry).
				For Coulomb systems, such unrestricted ground states are always bosonic (see, e.g.,~\cite[Cor.~3.1]{LiSe}).
				In the case of Helium-like systems $N=2$, the physical ground state is actually spatially symmetric.

		The present paper studies the relation between the $N$- and $(N-1)$-particle ground states in a two-cluster region in which one particle is far from the others.
		In the physics literature, one often assumes a separation of $N$-electron states $\psi_N$ such as
		\[
		\psi_N(x_1,\dots,x_N) \approx f(x_N)\phi_{N-1}(x_1,\dots,x_{N-1}) \quad (|x_i| \ll |x_N|\to\infty),
		\]
		where $f$ is some function, and $\phi_{N-1}$ denotes a physical ground state obeying the Pauli principle.
		Such an asymptotic factorization was discussed, e.g., in the book by Bethe and Salpeter~\cite{BeSa} for helium-like and hydrogen-like atoms, and heuristic arguments for general $N$-electron systems can be found in~\cite{KaDa, HH} and references therein.

		Mathematically, Ahlrichs et al. (AHOM) conjectured in~\cite[Eq.~(5.1)]{AHOM} that for the helium system
		\begin{equation}
			\label{conj.AHOM}
		\lim_{|x_2|\to\infty}\frac{\psi(x_1,x_2)}{\sqrt{\rho(x_2)}}= \phi(x_1)
		\end{equation}
		for every fixed $x_1$.
		In related work, Lieb and Simon~\cite{LS} proved that as $|x_2|\to \infty$
		\begin{equation}
			\label{eq.LS}
		\left|
		\frac{\psi(x_1',x_2)}{\psi(x_1,x_2)}-\frac{\phi(x_1')}{\phi(x_1)}
		\right|
		=o(1)
		\end{equation}
		uniformly for $x_1,x_1'$ in compact sets.
		We note that Lieb--Simon's result implies that there exists a function $\chi$ such that as $|x_2|\to \infty$ $$\frac{\psi(x_1,x_2)}{\chi(x_2)}- \phi(x_1)=o(1),$$ locally uniformly in $x_1$.
		In addition, Combes, M.~Hoffmann-Ostenhof, and T.~Hoffmann-Ostenhof showed~\cite{CHO} that \eqref{conj.AHOM} holds in $L^2$: as $|x_2|\to \infty$
		\begin{equation}
			\label{eq.CHO}
			\int_{\R^3}\left|\frac{\psi(x_1,x_2)}{\sqrt{\rho(x_2)}}-\phi(x_1)\right|^2\,dx_1=O\left(\frac{1}{|x_2|^2}\right).
			\end{equation}
			Their result was extended to $N$-particle ground states by Briet~\cite{Briet}.
			Combining \eqref{eq.LS} with~\eqref{eq.CHO}, we obtain the original AHOM conjecture~\eqref{conj.AHOM}.

			This paper is concerned with a quantitative refinement of the AHOM conjecture.
			Lieb and Simon already pointed out in~\cite{LS} that their convergence rate $o(1)$ could be improved by $o(|x_2|^{-1})$ in our notation.
			Formal analyses such as those in~\cite{BeSa,KaDa} suggest an error of order $|x_N|^{-2}$ when $X_{N-1}$ remains in a compact set.

		For $r_N:=|x_N|$, we introduce a region $A_{N,\alpha}=A_{N,\alpha}(r_N)$ by
		\[
		\begin{cases}
			A_{N,\alpha}(r_N) := \{X_{N-1}=(x_1,\dots,x_{N-1}) \colon |x_i| \le r_N^\alpha, \; \forall 1\le i\le N-1\}& (0<\alpha<1); \\
			 A_{N,0}(r_N):=\{X_{N-1}=(x_1,\dots,x_{N-1}) \colon |x_i| \le R_0\ln r_N, \; \forall 1\le i\le N-1\},
		\end{cases}
		\]
		where $R_0>0$ is a fixed constant.

		Our main result is the following.
		\begin{theorem}
			\label{thm.main}
			Suppose that $E_N\le E_{N-1}<E_{N-2}$ and that $E_N$ is an eigenvalue of $H_N$. Let $\psi$ and $\phi$ be the unique positive ground states of $H_N$ and $H_{N-1}$, respectively.
			Fix $0\le \alpha<1$, and define
			 \begin{equation}
			 \label{eq.fdef}
			f_\alpha(r):=\begin{cases}
				\ln r & (0<\alpha<1),\\
				(\ln r)^2 & (\alpha=0).
			\end{cases}
			\end{equation}
			As $|x_N|\to \infty$, we have
			\[
			\left|\frac{\psi(x_1,\dots,x_N)}{u(x_N)}-\phi(x_1,\dots,x_{N-1})\right| =O(r_N^{-2+\alpha}f_\alpha(r_N))
			\]
			uniformly on $A_{N,\alpha}(r_N)$.
			\end{theorem}

			\begin{remark}
				One logarithmic factor arises from $t=R\ln |x_N|$ in Lemma~\ref{lem.first} below.
				Because this is a technical choice, the logarithm loss might be removable.
			\end{remark}

			\begin{remark}
				According to the celebrated HVZ theorem~\cite{Teschl}, $E_N$ is a discrete eigenvalue when $E_N<E_{N-1}$.
				It is generally believed that $E_N<E_{N-1}$ guarantees $E_{N-1}<E_{N-2}$, but this is still open.
				Nevertheless, Zhislin's theorem~\cite{Zhislin,Teschl,LiSe} implies that $E_N<E_{N-1}<E_{N-2}$ whenever $N<Z+1$.
				Moreover, $E_{N-1}$ is an isolated eigenvalue of $H_{N-1}$ under our assumptions, which is used below.
				\end{remark}

			\begin{remark}
				\label{rem.threshold}
			    For certain critical charges $Z_c<N-1$, there exists a ground state at threshold $E_N=E_{N-1}$~\cite{HOS, BFLS,Gridnev}.
			    By contrast, there is no possibility for the existence of an unrestricted threshold ground state when $Z_c=N-1$~\cite{Goto}.
				\end{remark}

							\begin{theorem}
					\label{th.L2}
					Suppose that $E_N\le E_{N-1}<E_{N-2}$ and that $E_N$ is an eigenvalue of $H_N$.
					Then there exist constants $C,L_0>0$ such that for all $|x_N|\ge L_0$
					\[
					\int_{\R^{3N-3}}\left| \frac{\psi(x_1,\dots,x_N)}{\sqrt{\rho(x_N)}}-\phi(x_1,\dots,x_{N-1})\right|^2\,dx_1\cdots dx_{N-1}\le C|x_N|^{-2}.
					\]
				\end{theorem}

			Once these theorems are proved, we obtain the following quantitative form of the AHOM asymptotics.

		\begin{corollary}
			\label{Cor.main}
			Under the same assumptions of Theorem~\ref{thm.main}, as $|x_N|\to \infty$, we have
			\[
			\left|\frac{\psi(x_1,\dots,x_N)}{\sqrt{\rho(x_N)}}-\phi(x_1,\dots,x_{N-1})\right| =O(r_N^{-2+\alpha}f_\alpha(r_N))
			\]
			uniformly on $A_{N,\alpha}(r_N)$.
			\end{corollary}

		\begin{proof}[Proof of Corollary~\ref{Cor.main}]
		Combining Theorem~\ref{thm.main} with Theorem~\ref{th.L2}, we have
		\[
		2\left(1-
		\frac{u(x_N)}{\sqrt{\rho(x_N)}} \right)=\int_{\R^{3N-3}}\left| \frac{\psi(x_1,\dots,x_N)}{\sqrt{\rho(x_N)}}-\phi(x_1,\dots,x_{N-1})\right|^2\,dX_{N-1}=O(r_N^{-2})
		\]
		as $|x_N|\to \infty$.
		Since $\phi$ is bounded~\cite{Kato,AS}, the conclusion follows.
			\end{proof}

						\begin{remark}
						In the case of $E_N<E_{N-1}<E_{N-2}$, Theorem~\ref{th.L2} is known:~\cite{CHO} for two-electron atoms and~\cite{Briet} for general bosonic atoms.
						Combining this with Theorem~\ref{thm.main} gives Corollary~\ref{Cor.main}.
			The proofs in~\cite{CHO,Briet} rely on a partitioning argument from~\cite{CT} and use the binding condition $E_N<E_{N-1}$, so that their results do not cover the threshold case. In Section~\ref{sect.L2}, we give a different proof that is also applicable at threshold.
			\end{remark}

				We close this introduction with an outline of the proof strategy.

				The proof of Theorem~\ref{thm.main} is given in Section~\ref{sect.main}.
				The argument is based on Brownian motion ideas introduced by Lieb and Simon~\cite{LS}.
				Intuitively, the ground state $\psi$ can be approximated by $e^{-t(H_{N-1}-\Delta_N/2-E_N)}\psi$ for large separation $|x_i|\ll|x_N|\to\infty$ and long time $t\ll |x_N|$.
				In addition, $e^{-t(H_{N-1}-E_{N-1})}\psi$ approaches $u\phi$ as $t\to\infty$.
				Lieb and Simon~\cite{LS} considered general $N$-body systems and did not use the Coulombic nature except the decay rate.
			    In the remark following~\cite[Thm.~4.1]{LS}, the authors pointed out that for atomic systems their argument might be refined by replacing the contribution from the intercluster potentials $V_{ij}(x_i+b_i(t)-x_j-b_j(t))$ by a suitable quantity, instead of zero.
			    Our approach was motivated by this observation, and this replacement is one of the main ingredients of the proof.
				They also obtained the AHOM conjecture in the case of compactly supported potentials~\cite[Thm.~5.5]{LS}, and the error decays exponentially.
				For Coulomb interactions, our estimate is uniform on the expanding regions $A_{N,\alpha}$ and gives an error $o(r_N^{-1})$ for every $0\le\alpha<1$.
				In particular, it implies the improvement of the ratio asymptotics suggested in~\cite{LS}.

				We provide the proof of Theorem~\ref{th.L2} in Section~\ref{sect.L2}.
				There we use differential inequality techniques rather than the partitioning method employed in~\cite{CHO,Briet}.
				In~\cite[Sect.~II]{CHO}, they also obtained $1-u/\sqrt{\rho}=o(1)$ as $|x_N|\to\infty$ without the partitioning method.
				Our proof is a modification of their argument, and does not use Combes--Thomas type spectral analysis~\cite{CT,DHSV}; consequently, it does not require the binding condition $E_N<E_{N-1}$.
				It might be worth mentioning that the argument in Section~\ref{sect.L2} is independent of Section~\ref{sect.main}.

				Finally, our arguments rely heavily on upper and lower bounds of the one particle density~\eqref{eq.HOM} and~\eqref{eq.threshold} below.
				Although these can be obtained by Schr\"odinger inequality method as in~\cite{AHOM}, the lower bounds may not be explicitly written in literature for general $N$-particle cases.
				We include the proof for completeness in Appendix~\ref{App}.

			\section{Proof of Theorem~\ref{thm.main}}
			\label{sect.main}
			Throughout the proof, $C,c>0$ denote positive constants whose values may change from line to line.
			We also use $C_a$ to denote a positive constant depending on parameter $a$.
			These constants are independent of $X_N$, $t$, and $r_N$, but may depend on the fixed parameters $\alpha$, $Z$, $N$, etc.

			Let $U(x_N):=(Z-N+1)/r_N$ and $\widetilde V_N:=V_{N-1}-U(x_N)$.
			First, we approximate $\psi$ by $e^{-t(H_{N-1}-\Delta_N/2-U(x_N)-E_N)}\psi$ for large $|x_N|$.
			To do this, we define
			\[
			\widetilde{\psi}(X_N;t):=
			\E \left(\exp \left[\int_0^t (E_N-\widetilde V_N)(X_N+B(s)) \,ds\right]\psi(X_N+B(t))\right).
			\]
			Here $B(t)=(b_1(t),\dots,b_N(t))$, and $b_1,\dots,b_N$ are independent three-dimensional Brownian motions starting at the origin, where $\E$ denotes the Brownian expectation.
			For an event $A$, let $\mathbf1(A)=\mathbf1_A$ be its indicator and write $\mathbb{P}(A):=\E(\mathbf1_{A})$.
			Introducing $U$ is a new ingredient compared with the argument of Lieb and Simon~\cite{LS}.

			Throughout the present paper, we shall repeatedly use the following estimates.
			We will give the proofs of~\eqref{eq.HOM} and~\eqref{eq.threshold} in Appendix~\ref{App}.
			\begin{enumerate}
				\item Exponential bounds~\cite{AHO,DHSV,Morgan}: if $E_N<E_{N-1}$, then for every
				$0<\sigma<\sqrt{2\epsilon_N}$, there exists
				$C_\sigma>0$ such that
				\begin{equation}
					\label{eq.Agmon}
					\psi(X_N)
					\le C_\sigma e^{-\sigma|X_N|},
					\qquad
					|X_N|
					:=\left(\sum_{i=1}^N|x_i|^2\right)^{1/2},
				\end{equation}
				where $\epsilon_N:=E_{N-1}-E_N$ is the ionization energy.
			See~\cite{Agmon,CS} for the standard Agmon estimates.

			\item The asymptotics of the ground state density when $E_N<E_{N-1}$: for sufficiently large $r=|x|>0$, it holds that
			\begin{equation}
				\label{eq.HOM}
				cr^{(Z-N+1)/\sqrt{2\epsilon_N}-1} e^{-\sqrt{2\epsilon_N} r}\le u(x)\le\sqrt{\rho(x)} \le Cr^{(Z-N+1)/\sqrt{2\epsilon_N}-1} e^{-\sqrt{2\epsilon_N} r}
				\end{equation}
				where $0<c<C$ are some constants.
				The upper bound was given in~~\cite{AHOM}, and the general lower bound will be proved in Appendix~\ref{App}.

				\item For the threshold ground state, the preceding bounds are replaced by
				\begin{equation}
					\label{eq.threshold}
					c r^{-3/4} e^{-\kappa_N\sqrt{r}}\le u(x)\le\sqrt{\rho(x)} \le C r^{-3/4} e^{-\kappa_N\sqrt{r}},
				\end{equation}
				where $\kappa_N:=\sqrt{8\left(N-1-Z\right)}$.
				For the two-electron atom, this was already pointed out in~\cite[Remark~4(2)]{HOS}, but their bound allows an arbitrary $\delta>0$ loss in the polynomial exponents.
				The same comparison argument gives the same form for $Z_c<N-1$.
			\end{enumerate}
			The following lemma is a refinement of the counterpart in~\cite{LS} and provides the leading order of Theorem~\ref{thm.main}.

			\begin{lemma}
				\label{lem.first}
				Fix $0\le\alpha <1$.
				Then, for $R>1$ and $0<t\le R \ln r_N$, we have
				\[
				\sup_{X_{N-1}\in A_{N,\alpha}}{|\psi(X_N)-\widetilde{\psi}(X_N;t)| }\le C_{R}\frac{f_\alpha(r_N)}{r_N^{2-\alpha}}u(x_N),
				\]
				for sufficiently large $|x_N|$, where $f_\alpha$ is given in~\eqref{eq.fdef}.
				\end{lemma}

				\begin{proof}
					By the Feynman--Kac formula, we see
					\begin{equation}
						\label{eq.FKequ}
						\begin{split}
						|\psi(X_N)-\widetilde \psi(X_N;t)|
						&=\left|\E\left(\psi(X_N+B(t)) e^{F-G}[e^G -1]\right)\right|,
						\end{split}
						\end{equation}
						where $F,G$ are defined by
						\begin{align}
						F&:=\int_0^t(E_N-V_N)(X_N+B(s))\,ds,\nonumber\\
						G&:=\int_0^t \left[\frac{N-1}{|x_N+b_N(s)|}-\sum_{i=1}^{N-1}\frac{1}{|x_i+b_i(s)-(x_N+b_N(s))|}\right]\,ds.\label{eq.Gint}
						\end{align}
						First, we estimate the contribution from the event $\{\sup_{s\le t}|b_{N}(s)|\ge \theta r_N\}$ for small $0<\theta<1$.
						Together with $\mathbb{P}(|b(t)|\ge A)\le Ce^{-cA^2/t}$ from~\cite[Eq.~(3.4')]{Simon} and L\'evy's maximal inequality~\cite[Thm.~3.6.5]{Simon}, we have
						\begin{equation}
							\label{eq.Bigb}
							\begin{split}
								&\E\left(\psi e^{F-G}|e^G -1|\mathbf1\left(\sup\limits_{s\le t}|b_{N}(s)|\ge \theta r_N\right) \right)\\
								&\le
								C\E\left(|e^{F-G}[e^G -1]|^2\right)^{1/2} \mathbb{P}(|b_{N}(t)|\ge c\theta r_N )^{1/2}\\
									&\le
									C\E\left(|e^{F-G}[e^G -1]|^2\right)^{1/2}\exp\left(-c_\theta\frac{r_N^2}{t}\right),
								\end{split}
							\end{equation}
							where we have also used $\|\psi\|_\infty<\infty$~\cite{Kato,AS}.
						 Using the fact that $\E(e^{a\int_0^t |x_i+b_i(s)|^{-1} \,ds})\le C_ae^{c_a t}$~\cite{Carmona}, the right-hand side in~\eqref{eq.Bigb} is bounded by
						 \begin{equation}
						 	\label{eq.Bigb2}
						 	\begin{split}
						 	\E\left(\psi e^{F-G}|e^G -1| \mathbf1\left(\sup_{s\le t}|b_{N}(s)|\ge \theta r_N\right)\right)
						 	&\le
						 	C\exp\left(Ct -c_\theta\frac{r_N^2}{t}\right)\\
						 	&\le
						 	C_{\theta,R}u(x_N)\exp(-c_{\theta,R}r_N)
						 	\end{split}
						 	\end{equation}
						 	for large $|x_N|$, where we have used the asymptotics~\eqref{eq.HOM},~\eqref{eq.threshold} and $r^{-a}\ln r\to 0$ for any $a>0$ as $r\to \infty$.

						 	To obtain the desired bound from the complementary event, we note from the subsolution estimate~\cite{AS,KalfHinz} that for any $X_N\in\R^{3N}$
						 	\begin{equation}
						 		\label{eq.subsol}
						 		\begin{split}
						 			\psi(X_N)
						 			&\le C\left(\int_{|X_N-Y_N|<1} \psi(Y_N)^2\, dY_N\right)^{1/2}
						 			\le C\left(\int_{|x_N-y_N|<1}\rho(y_N) \,dy_N\right)^{1/2} \\
						 			&\le
						 			C\sqrt{\rho(x_N)},
						 		\end{split}
						 	\end{equation}
						 	where we have used $\sup_{|x-y|<1}[\rho(y)/\rho(x)]\le C$ by the asymptotics~\eqref{eq.HOM} and~\eqref{eq.threshold}.
						 	Using the asymptotics~\eqref{eq.HOM} and~\eqref{eq.threshold} again, for $|b_N(t)|\le \theta r_N$, we have
						 	\[
						 	\psi(X_N+B(t))\le C\sqrt{\rho(x_N+b_N(t))}\le C_\theta\sqrt{\rho(x_N)}e^{c_N |b_N(t)|}\le C_\theta u(x_N)e^{c_N |b_N(t)|},
						 	\]
						 	where $c_N>0$ is certain constant.
						 	Then we infer that for sufficiently large $L=L(R)>0$
						 	\begin{equation}
						 		\label{eq.smallb}
						 		\begin{split}
						 			&\left|\E\left(\psi e^{F-G}[e^G -1]\mathbf1\left(\sup\limits_{s\le t}|b_{N}(s)|\le \theta r_N\right) \mathbf1\left(\max\limits_{1\le i\le N-1} \sup\limits_{s\le t}|b_i(s)|\ge L \ln r_N\right)\right)\right|
						 			\\
						 			&\le
						 			C_\theta u(x_N)\E\left(e^{2c_N |b_N(t)|}|e^{F-G}[e^G -1]|^2\right)^{1/2}\exp\left(-\frac{cL^2}{R} \ln r_N\right)\\
						 			&\le
						 			C_{\theta,R}\frac{u(x_N)}{r_N^{\beta_{L,R}}},
						 		\end{split}
						 	\end{equation}
						 	where we have used $\E(e^{a|b_N(t)|})\le C_a e^{c_a t}$ and chosen $L>0$ so large that $\beta_{L,R}>2$.

						 		Combining~\eqref{eq.Bigb2} with~\eqref{eq.smallb}, we arrive at
						\begin{equation}
							\label{eq.FKest}
							\begin{split}
						&|\psi(X_N)-\widetilde \psi(X_N;t)|\\
						&\le
						\left|
						\E\left(\psi e^{F-G}[e^G -1]\mathbf1\left(\sup\limits_{s\le t}|b_{N}(s)|\le \theta r_N\right) \mathbf1\left(\max_{1\le i\le N-1} \sup_{s\le t}|b_i(s)|\le L \ln r_N\right)\right)
						\right|\\
						&\quad+C_{\theta,R}\frac{u(x_N)}{r_N^{\beta_{L,R}}}
						\end{split}
						\end{equation}
						for sufficiently large $L$.

						Our next task is to estimate the first term on the right-hand side in~\eqref{eq.FKest}.
						Taylor's formula gives
						\[
						\frac{1}{|x-y|}
						=
						\frac{1}{|x|}+\frac{x\cdot y}{|x|^3}
						+O\left(\frac{|y|^2}{|x|^{3}}\right)
						\]
						for $x,y \in \R^3$ with $|y|\le \sigma|x|$ $(0<\sigma<1)$.
						Now we consider the case $0<\alpha<1$.
						On the event of $\sup_{s\le t}|b_{N}(s)|\le \theta r_N$ and $\max_{1\le i\le N-1} \sup_{s\le t}|b_i(s)|\le L \ln r_N$, the integrand in~\eqref{eq.Gint} admits the following estimate
						\begin{equation}
							\label{eq.Taylor}
							\begin{split}
								&\frac{N-1}{|x_N+b_N(s)|}-\sum_{i=1}^{N-1}\frac{1}{|x_i+b_i(s)-(x_N+b_N(s))|}
								\\&=
								-\sum_{i=1}^{N-1}\frac{(x_i+b_i(s))\cdot (x_N+b_N(s))}{|x_N+b_N(s)|^3}
								+O\left(r_N^{-3+2\alpha}\right)\\
								&=
								O\left(r_N^{-2+\alpha} \right).
								\end{split}
								\end{equation}
					Consequently, on this event,
					\begin{equation}
						\label{eq.lem1G}
						\begin{split}
							\left|e^G-1\right|
							\le
							Ctr_N^{-2+\alpha}
							\le C_R\frac{\ln r_N}{r_N^{2-\alpha}}.
							\end{split}
							\end{equation}
							Finally, we note $e^{F-G}|e^G-1|\le C_Re^F|e^G-1|$ on this event by~\eqref{eq.Taylor}.
							Combining~\eqref{eq.FKest} and~\eqref{eq.lem1G} with $\E(\psi(X_N+B(t)) e^F)=\psi (X_N)\le Cu$, we complete the proof for $0<\alpha<1$.

							When $\alpha=0$, the only change required is to replace $O(r_N^{-2+\alpha})$ in~\eqref{eq.Taylor} with $O(r_N^{-2} \ln r_N)$.
							This proves the assertion for $\alpha=0$ as well.
					\end{proof}

The following lemma is essentially the same as~\cite[Lemma~3.2]{LS}, but we include the proof for completeness.

\begin{lemma}
	\label{lem.proj}
	For every $t\ge 1$, $f\in L^2(\R^{3N-3})$, and $X_{N-1}\in\R^{3N-3}$, one has
	\[
	\left|
	(e^{-tH_{N-1}} f)(X_{N-1})
	-(\phi, f)_{L^2(\R^{3N-3})} e^{-tE_{N-1}} \phi(X_{N-1})
	\right|
	\le
	Ce^{-tE_{N-1}^{(1)}}\|f\|_{L^2(\R^{3N-3})},
	\]
	where $E^{(1)}_{N-1}:=\inf \mathrm{spec}(H_{N-1}\upharpoonright\{\phi\}^{\perp})$ stands for the first  excited energy of $H_{N-1}$, and $E_{N-1}<E^{(1)}_{N-1}$.
	\end{lemma}

	\begin{proof}
		Let $P=|\phi\rangle \langle \phi| $ be the projection onto $\phi$, and let $Q:=1-P$.
		Since $e^{-H_{N-1}}$ is bounded from $L^2$ to $L^\infty$~\cite[Theorem~B.1.1]{SimonS}, we see
		\begin{align*}
			\|e^{-tH_{N-1}} Q f \|_{\infty}
			&\le
			\|e^{-H_{N-1}}  \|_{L^2\to L^\infty}
			\|e^{-(t-1)H_{N-1}} Q f \|_{L^2}\\
			&\le
			\|e^{-H_{N-1}}  \|_{L^2\to L^\infty}
			e^{-(t-1)E_{N-1}^{(1)}}
			\|Q f \|_{L^2}\\
			&\le
			Ce^{-tE_{N-1}^{(1)}}\|f \|_{L^2}.
			\end{align*}
			By definition, we have $e^{-tH_{N-1}} Qf =e^{-tH_{N-1}} f-(\phi, f)e^{-tE_{N-1}}\phi$, which shows the desired result.
		\end{proof}

		Let $\E_N$ denote a Brownian expectation with respect to $b_N$ and define
		\[
		\chi(x_N;t):=
		\E_N \left(\exp\left[\int_0^tU(x_N+b_N(s))\,ds\right]  u(x_N+b_N(t)) \right),
		\]
		where $u$ is given in~\eqref{def.u}.

		 We next approximate $\widetilde \psi$  in the following sense.

		\begin{lemma}
			\label{lem.PtoFK}
		Fix $0\le\alpha <1$ and $R>0$, and let $\mu_R:=R(E_{N-1}^{(1)}-E_{N-1})>0$.
		Then, for all sufficiently large $|x_N|$
		\[
		\sup_{X_{N-1}\in A_{N,\alpha}}	\left|
		\widetilde\psi(X_{N};t) -e^{-t\epsilon_N}\chi(x_N;t)\phi(X_{N-1})
		\right|
		\le
		Cr_N^{-\mu_R}e^{-t\epsilon_N}\chi(x_N;t),
		\]
		where $t=R \ln r_N$ and $\epsilon_N=E_{N-1}-E_N$.
		\end{lemma}

		\begin{proof}
			Let $P=|\phi\rangle\langle\phi|$ and $Q=1-P$, and write
			$\psi_y:=\psi(\,\cdot\,,y)$.
			Since $P\psi_y=u(y)\phi$, Lemma~\ref{lem.proj} gives
			\begin{equation}
				\label{eq.PFK}
				\begin{split}
					&\left|
			\widetilde\psi(X_{N};t) -e^{-t\epsilon_N}\chi(x_N;t)\phi(X_{N-1})
			\right|\\
			&=
			\left|
			\E_N\left(e^{\int_0^tU(x_N+b_N(s))\,ds}(e^{-t(H_{N-1} -E_N)}Q \psi_{x_N+b_N(t)})(X_{N-1})\right)
			\right|
			\\
			&\le
			Ce^{-t(E_{N-1}^{(1)} -E_N)}\E_N(e^{\int_0^tU(x_N+b_N(s))\,ds} \|\psi_{x_N+b_N(t)}\|_{2}).
			\end{split}
			\end{equation}
			We note that $\sqrt{\rho(x_N)}\le Cu(x_N)$ for all $x_N$ from~\eqref{eq.HOM},~\eqref{eq.threshold} and continuity of the strictly positive $u$ and $\rho$.
			Hence we see
			\begin{equation}
				\label{eq.PFK2}
				\begin{split}
			&e^{-t(E_{N-1}^{(1)} -E_N)}\E_N(e^{\int_0^tU(x_N+b_N(s))\,ds} \|\psi_{x_N+b_N(t)}\|_{2})\\
			&=e^{-t \epsilon_N}e^{-t(E_{N-1}^{(1)} -E_{N-1})}\E_N(e^{\int_0^tU(x_N+b_N(s))\,ds} \sqrt{\rho(x_N+b_N(t))})\\
			&\le
			Cr_N^{-\mu_R}e^{-t \epsilon_N} \chi(x_N;t),
			\end{split}
			\end{equation}
			which shows the lemma.
			\end{proof}
			We choose $R>1$ so that $\mu_R>2$, and let $\eta(x_N):=e^{-\epsilon_NR\ln r_N } \chi(x_N;R\ln r_N)$ for brevity.
			Now we can state the following pointwise approximation.

			\begin{lemma}
				\label{lem.PC}
				For fixed $0\le\alpha <1$, we have
				\[
				\sup_{X_{N-1}\in A_{N, \alpha}}\left|
				{\psi(X_N)}-\eta(x_N)\phi(X_{N-1})
				\right|
				\le
				C\frac{f_\alpha(r_N)}{r_N^{2-\alpha}}(u(x_N)+\eta(x_N))
				\]
				for sufficiently large $|x_N|$.
				\end{lemma}

				\begin{proof}
					Set $t=R\ln r_N$.
					Combining Lemma~\ref{lem.first} with Lemma~\ref{lem.PtoFK}, we have
					\begin{equation}
						\label{eq.com}
						\begin{split}
							&\sup_{X_{N-1}\in A_{N, \alpha}}\left|
							{\psi(X_N)}-\eta(x_N)\phi(X_{N-1})
							\right|\\
							&\le
							\sup_{X_{N-1}\in A_{N, \alpha}}\left|
						\psi(X_N)-\widetilde\psi(X_N;t)
							\right|+
							\sup_{X_{N-1}\in A_{N, \alpha}}\left|
							\widetilde\psi(X_N;t)-\eta(x_N)\phi(X_{N-1})
							\right|\\
							&\le
							C\left(\frac{f_\alpha(r_N)}{r_N^{2-\alpha}}u(x_N)
							+\frac{\eta(x_N)}{r_N^{\mu_R}}\right),
							\end{split}
							\end{equation}
							where we have taken $\mu_R>2$. This yields the desired result.
					\end{proof}
					It remains to replace $\eta$ by $u$.

					\begin{lemma}
						\label{lem.appro.u}
						Fix $0\le \alpha<1$.
						For sufficiently large $|x_N|$, we have
						\[
						\left|
						\eta(x_N)-u(x_N)
						\right|
						\le C \frac{f_\alpha(r_N)}{r_N^{2-\alpha}}u(x_N).
						\]
						\end{lemma}

						\begin{proof}
							First, we consider the case $0<\alpha<1$.
								Using the asymptotics~\eqref{eq.Agmon} for $\phi$, we see $\phi(X_{N-1})\le C_\sigma\exp(-\sigma|X_{N-1}|)$ and deduce from the Cauchy--Schwarz inequality that
							\begin{equation}
								\label{eq.ubound}
								\begin{split}
								\int_{|x_i|\ge r_N^\alpha} \psi(X_N)\phi(X_{N-1})\,dx_1\cdots dx_{N-1}
								&\le
								C\sqrt{\rho(x_N)}\exp(-c_\sigma r_N^\alpha),
								\end{split}
								\end{equation}
								for some $c_\sigma>0$.
								Using this and the exponential decay of $\phi$, we have
								\begin{equation}
									\label{eq.ueta}
									\begin{split}
										|\eta(x_N)-u(x_N)|
										&=
										\left|
										\int_{\R^{3N-3}}\phi(X_{N-1})\left(\eta(x_N) \phi(X_{N-1}) -\psi(X_N)\right)\,dX_{N-1}
										\right|\\
										&\le
										\int_{A_{N,\alpha}} \phi(X_{N-1})\left|\eta(x_N) \phi(X_{N-1}) -\psi(X_N)\right|\,dX_{N-1}\\
										&\quad+\sum_{i=1}^{N-1}\int_{|x_i|\ge r_N^\alpha}  \phi(X_{N-1})\left|\eta(x_N) \phi(X_{N-1}) -\psi(X_N)\right|\,dX_{N-1}\\
										&\le
										\int_{A_{N,\alpha}} \phi(X_{N-1})\left|\eta(x_N) \phi(X_{N-1}) -\psi(X_N)\right|\,dX_{N-1}\\
										&\quad+
										C\left(\eta (x_N)\exp(-2c_\sigma r_N^\alpha)+\sqrt{\rho(x_N)}\exp(-c_\sigma r_N^\alpha)\right).
									\end{split}
								\end{equation}
								For the first term on the right-hand side in~\eqref{eq.ueta}, Lemma~\ref{lem.PC} implies that
								\begin{equation}
									\label{eq.appLempc}
									\begin{split}
										\int_{A_{N,\alpha}} \phi(X_{N-1})\left|\eta(x_N) \phi(X_{N-1}) -\psi(X_N)\right|\,dX_{N-1}
										&\le
										Cf_\alpha(r_N)\frac{\eta(x_N)+u(x_N)}{r_N^{2-\alpha}},
										\end{split}
										\end{equation}
										where we have used $\phi \in L^1(\R^{3N-3})$.
										Together with~\eqref{eq.ueta} and~\eqref{eq.appLempc}, we arrive at
										\begin{equation}
											\label{eq.conc}
										|\eta(x_N)-u(x_N)|
										\le
										Cf_\alpha(r_N)\frac{\eta(x_N)+u(x_N)}{r_N^{2-\alpha}}+Cu(x_N)\exp(-cr_N^\alpha),
										\end{equation}
										where we have used $\sqrt{\rho(x_N)}\le C u(x_N)$.
										This also implies that
										\[
										\eta(x_N)\le |\eta(x_N)-u(x_N)| +u(x_N)\le Cr_N^{-1}\eta(x_N)+Cu(x_N),
										\]
										and thus $\eta(x_N)\le Cu(x_N)$ for sufficiently large $|x_N|$.
										Substituting this estimate into~\eqref{eq.conc} proves the desired result for $0<\alpha<1$.

										For  $\alpha=0$, we replace $A_{N,0}$ with $\widetilde A_{N,0}:=\{|x_i|\le \widetilde R_0 \ln r_N, \,\forall 1\le i\le N-1\}$ for sufficiently large $\widetilde R_0>R_0$.
										Replacing $\exp(-c_\sigma r_N^\alpha)$ in~\eqref{eq.ubound}--\eqref{eq.conc} by $r_N^{-c_\sigma\widetilde R_0}$ and choosing $c_\sigma\widetilde R_0>2$, we obtain the assertion for $\alpha=0$.
										Therefore, the proof is complete.
							\end{proof}

							We are now ready to prove the main theorem.
							\begin{proof}[Proof of Theorem~\ref{thm.main}]
								Combining Lemma~\ref{lem.PC} and Lemma~\ref{lem.appro.u}, we have
								\begin{equation}
									\label{eq.con}
									\begin{split}
										\sup_{X_{N-1}\in A_{N, \alpha}}\left|
										{\psi(X_N)}-u(x_N)\phi(X_{N-1})
										\right|
										&\le
										\sup_{X_{N-1}\in A_{N, \alpha}}\left|
										{\psi(X_N)}-\eta(x_N)\phi(X_{N-1})
										\right|\\
										&\quad+
										\|\phi\|_\infty\left|
										\eta(x_N)
										-u(x_N)
										\right|\\
									&\le
									C\frac{f_\alpha(r_N)}{r_N^{2-\alpha}}u(x_N),
									\end{split}
									\end{equation}
									where we have used the boundedness of $\phi$~\cite{Kato,AS} and ${\eta}(x_N)\le Cu(x_N)$ as established in the proof of Lemma~\ref{lem.appro.u}.
									Dividing by $u(x_N)>0$ completes the proof.
							\end{proof}

							\section{Proof of Theorem~\ref{th.L2}}
							\label{sect.L2}
							In this section, we prove Theorem~\ref{th.L2}.
							Following~\cite[Lemma~2.2]{CHO}, we obtain a Schr\"odinger inequality for $\sqrt{\rho}$.

							\begin{lemma}
								\label{lem.Schroineq}
								For sufficiently large $|x_N|$, we have
								\begin{align*}
									\left(-\frac{\Delta_N}{2}-\frac{Z-N+1}{r_N}+\epsilon_N+\frac{E_{N-1}^{(1)}-E_{N-1}}{2}\left(1-\frac{u^2}{\rho}\right)\right)
									\sqrt{\rho}
									\le C\frac{\sqrt{\rho}}{r_N^{2}}
									\end{align*}
									in the distributional sense, where $E_{N-1}^{(1)}$ is the first excited energy of $H_{N-1}$ defined in Lemma~\ref{lem.proj}.
								\end{lemma}
								\begin{proof}
									The Schr\"odinger equation $(H_{N}-E_N)\psi=0$ yields
									\begin{align*}
										0&=\int_{\R^{3N-3}}\psi \left(H_{N-1}-E_{N-1}  \right)\psi\,dx_1\cdots{dx_{N-1}}\\
										&\quad+\int_{\R^{3N-3}}\psi \left(-\frac{\Delta_N}{2}-\frac{Z}{r_N}+\sum_{i=1}^{N-1}\frac{1}{|x_i-x_N|}+\epsilon_N  \right)\psi\,dx_1\cdots{dx_{N-1}},
										\end{align*}
										in the distributional sense with respect to $x_N$.
										Let $P=|\phi\rangle \langle \phi|$ and $Q=1-P$ be the projection as in Lemma~\ref{lem.proj}.
										By definition, we see
										\[
										H_{N-1}-E_{N-1} \ge \left(E_{N-1}^{(1)}-E_{N-1}\right)Q
										\]
										in the quadratic form sense.
										Then we infer that
										\begin{equation}
											\label{eqA.1}
											\int_{\R^{3N-3}}\psi \left(H_{N-1}-E_{N-1}  \right)\psi\,dx_1\cdots{dx_{N-1}}
											\ge \left(E_{N-1}^{(1)}-E_{N-1}\right)\left(\rho-u^2\right).
											\end{equation}
											Using the Hoffmann-Ostenhof inequality~\cite{HH}
											\[
											-\sqrt{\rho}\Delta_N \sqrt{\rho}\le
											-\int_{\R^{3N-3}}\psi(X_N)\Delta_N \psi(X_N)\,dx_1\cdots dx_{N-1},
											\]
											we have
											\begin{equation}
												\label{eqA.Sch1}
												\begin{split}
											0&\ge
											\sqrt{\rho}\left(-\frac{\Delta_N}{2}-\frac{Z}{r_N}+\epsilon_N\right)\sqrt{\rho}+
											\left(E_{N-1}^{(1)}-E_{N-1}\right)\left(\rho-u^2\right)\\
											&\quad+\sum_{i=1}^{N-1}\int_{\R^{3N-3}}\frac{|\psi(X_N)|^2}{|x_i-x_N|}\,dx_1\cdots{dx_{N-1}}.
											\end{split}
											\end{equation}
											Now $(f,g)$ denotes the $L^2(\R^{3N-3})$ inner product with respect to $X_{N-1}$.
											By a direct calculation, we have for any real-valued $W$
											\begin{align*}
												( \psi, W \psi)
												&=u^2(\phi,W\phi)
												+2u\mathrm{Re}\,(\phi,W Q\psi)+(Q\psi,W Q\psi).
												\end{align*}
												Applying this to $W=\sum_i|x_i-x_N|^{-1}$ implies that, with $V=(N-1)/r_N$ and $F:=-V+W$,
												\begin{align*}
													\sum_{i=1}^{N-1}\int_{\R^{3N-3}}\frac{|\psi(X_N)|^2}{|x_i-x_N|}\,dx_1\cdots{dx_{N-1}}
													&\ge
													u^2(\phi,W\phi) -2u\|F\phi\|_2 \|Q\psi\|_{L^2(\R^{3N-3})},
													\end{align*}
													where we have used $(\phi, Q\psi)=0$.
											Now we  split
											\begin{align*}
											&\int_{\R^{3N-3}}\left(\frac{1}{|x_i-x_N|}-\frac{1}{|x_N|} \right)^2 \phi(X_{N-1})^2\,dX_{N-1}\\
											&=\int_{2|x_i|<r_N}\left(\frac{1}{|x_i-x_N|}-\frac{1}{|x_N|} \right)^2 \phi(X_{N-1})^2\,dX_{N-1}\\
											&\quad+\int_{{2|x_i|>r_N}}\left(\frac{1}{|x_i-x_N|}-\frac{1}{|x_N|} \right)^2 \phi(X_{N-1})^2\,dX_{N-1}.
											\end{align*}
											From the exponential decay of $\phi$, we see
											\[
											\int_{{2|x_i|>r_N}}\left(\frac{1}{|x_i-x_N|}-\frac{1}{|x_N|} \right)^2 \phi(X_{N-1})^2\,dx_1\cdots dx_{N-1}
											\le C\exp(-c|x_N|),
											\]
											where we have used the local integrability of $|x|^{-2}$.
											On $\{2|x_i|<r_N\}$, we use
											$$\left |\frac{1}{|x_i-x_N|}-\frac{1}{|x_N|}\right|
											=\frac{\left|r_N-|x_i-x_N|\right|}{r_N|x_i-x_N|}
											\le
											\frac{|x_i|}{r_N(r_N-|x_i|)}
											\le
											 \frac{2|x_i|}{r_N^2}.$$
											 Together with $\||x_i|\phi\|_2<\infty$, we obtain
											\begin{equation}
												\label{eqA.phi}
												\begin{split}
											\int_{\R^{3N-3}}\left(\frac{1}{|x_i-x_N|}-\frac{1}{|x_N|} \right)^2 \phi(X_{N-1})^2\,dx_1\cdots dx_{N-1}
												&\le Cr_N^{-4}.
												\end{split}
												\end{equation}
												Summing over $i$, we deduce from the Cauchy--Schwarz inequality that $\|F \phi\|_2 \le C/r_N^{2}$ and thus
												\begin{equation}
													\label{eqA.Wpsi}
												\sum_{i=1}^{N-1}\int_{\R^{3N-3}}\frac{|\psi(X_N)|^2}{|x_i-x_N|}\,dX_{N-1}
												\ge
												\frac{N-1}{|x_N|}u(x_N)^2-\frac{C}{r_N^2}u(x_N)^2-C\frac{u(x_N)\sqrt{\rho(x_N)-u(x_N)^2}}{r_N^2},
												\end{equation}
												where we have used $\|Q\psi\|_2^2 = \rho-u^2$.

											Dividing~\eqref{eqA.Sch1} by $\sqrt{\rho}$ and using~\eqref{eqA.Wpsi} leads to
											\begin{equation}
												\label{eqA.2}
												\begin{split}
												\frac{C}{r_N^2}\sqrt{\rho}&\ge
												\left(-\frac{\Delta_N}{2}-\frac{Z-N+1}{r_N}+\epsilon_N\right)\sqrt{\rho}
												\\&
												\quad+\left(E_{N-1}^{(1)}-E_{N-1}-\frac{N-1}{|x_N|}\right)\left(1-\frac{u^2}{\rho}\right)\sqrt{\rho}
												\end{split}
												\end{equation}
												where we have used $u\le \sqrt{\rho}$ and $0\le 1-u^2/\rho \le 1$.
												Since $E_{N-1}^{(1)}-E_{N-1}>0$, we have
												\[
												\left(E_{N-1}^{(1)}-E_{N-1}-\frac{N-1}{|x_N|}\right)\left(1-\frac{u^2}{\rho}\right)\ge \frac{E_{N-1}^{(1)}-E_{N-1}}{2}\left(1-\frac{u^2}{\rho}\right)
												\]
												for $|x_N|$ large enough, which shows the conclusion.
									\end{proof}

									 The following is a simple estimate on the effective potential.
									\begin{lemma}
										\label{lemA.ubound}
										\begin{equation}
											\label{eqA.ubound}
											\left|\left(-\frac{\Delta_N}{2}-\frac{Z-N+1}{r_N}+\epsilon_N\right)u\right|
											\le
											C\frac{\sqrt{\rho}}{r_N^{2}}
										\end{equation}
										\end{lemma}

										\begin{proof}
											The integration by parts in the Schr\"odinger equation $H_N\psi=E_N\psi$ leads to
											\[
											\left(-\frac{\Delta_N}{2}-\frac{Z-N+1}{r_N}+\epsilon_N\right)u
											=\frac{N-1}{r_N}u-\sum_{i=1}^{N-1}\int_{\R^{3N-3}}\frac {\phi(X_{N-1})\psi(X_N)}{|x_i-x_N|}\,{dX_{N-1}}.
											\]
											Using~\eqref{eqA.phi} again, we have
											\[
											\left|\sum_{i=1}^{N-1}\int_{\R^{3N-3}}\left(\frac {1}{|x_i-x_N|}-\frac{1}{r_N}\right) \phi(X_{N-1})\psi(X_N)\,{dX_{N-1}}
											\right|
											\le \|F\phi\|_2\sqrt{\rho}\le \frac{C\sqrt{\rho}}{r_N^2},
											\]
											which is the desired result~\eqref{eqA.ubound}.
											\end{proof}

											\begin{remark}
												Lemma~\ref{lem.Schroineq} and Lemma~\ref{lemA.ubound} are independent of the asymptotics for $\sqrt{\rho}$ and $u$, and will be used in Appendix~\ref{App} below.
												\end{remark}
												Now we turn to the

									\begin{proof}[Proof of Theorem~\ref{th.L2}]
										We note that
										\[
										\left(1-\frac{u^2}{\rho}\right)\sqrt{\rho}
										=(\sqrt{\rho}-u)\left(1+\frac{u}{\sqrt{\rho}}\right)
										\ge \sqrt{\rho}-u=:w.
										\]
										Let us introduce a differential operator $\widehat{L}$ by
										\[
										\widehat{L}:=-\frac{\Delta_N}{2}-\frac{Z-N+1}{r_N}+\epsilon_N+\frac{E_{N-1}^{(1)}-E_{N-1}}{2}.
										\]
										Then Lemma~\ref{lem.Schroineq}, Lemma~\ref{lemA.ubound}, and $\sqrt{\rho}\le Cu$ imply that
										\begin{equation}
											\label{eqA.wbound}
											\begin{split}
											\widehat{L}w
											\le
											C\frac{u}{r_N^{2}}+
											\left|
											\left(-\frac{\Delta_N}{2}-\frac{Z-N+1}{r_N}+\epsilon_N\right)u
											\right|
											\le C_0\frac{u}{r_N^{2}}.
											\end{split}
											\end{equation}
											Moreover, $u$ is a solution of the one-particle Schr\"odinger equation $-\Delta u/2+q(x)u=0$ with $q$ defined by
											\[
											q(x):=\epsilon_N -\frac{Z-N+1}{|x|}+\frac{E_{N-1}^{(1)}-E_{N-1}}{2}-\frac{\widehat{L} u(x)}{u(x)}.
											\]
											By Lemma~\ref{lemA.ubound} and  $\sqrt{\rho}\le Cu$, this $q$ is bounded for large $|x_N|$.
											From the Harnack inequality~\cite{AS,SimonS} and gradient estimate~\cite[Thm.~C.2.5]{SimonS}, we have
											\[
											\sup_{|y-x|\le 1}|\nabla u(y)|
											\le C\sup_{|y-x|\le 2} u(y)\le Cu(x)
											\]
											for large $|x|=r$.
											Then it follows that
											\begin{align*}
												-\Delta (r^{-2}u)=-r^{-2}\Delta u -2(\nabla r^{-2})\cdot \nabla u -u\Delta r^{-2}
												\ge
												-r^{-2}\Delta u -C\frac{u}{r^3}.
												\end{align*}
												Together with~\eqref{eqA.ubound}, we can take $r_0$ so that for all $r_N\ge r_0$
												\begin{equation}
													\label{eqA.rubound}
													\begin{split}
													\left(-\frac{\Delta_N}{2}-\frac{Z-N+1}{r_N}+\epsilon_N+\frac{E_{N-1}^{(1)}-E_{N-1}}{2}\right)\frac{u}{r_N^2}
													\ge
													\underbrace{\frac{E_{N-1}^{(1)}-E_{N-1}}{4}}_{=:c_0}\frac{u}{r_N^2}.
													\end{split}
													\end{equation}

													Now we also choose $r_0$ large enough so that for any $r_N\ge r_0$
													\[
													-\frac{Z-N+1}{r_N}+\epsilon_N+\frac{E_{N-1}^{(1)}-E_{N-1}}{2}>0.
													\]
													Taking $A_0 >0$ such that
													\[
												\max\left\{\frac{C_0}{c_0}, r_0^2\sup_{|x|=r_0}\frac{w(x)}{u(x)}\right\}\le A_0,
													\]
													together with~\eqref{eqA.wbound} and~\eqref{eqA.rubound}, we obtain
													\[
													\widehat{L}\left(w-A_0\frac{u}{r_N^2}\right)
													\le
													(C_0-A_0c_0)\frac{u}{r_N^2} \le 0.
													\]
													Clearly, $0\le w \le \sqrt{\rho}\to 0$ as $|x_N|\to\infty$.
													Hence, we deduce from the maximum principle that
													\[
													\sup_{|x_N|\ge r_0}\left(w-A_0\frac{u}{r_N^2}\right)\le
													\max \left\{\sup_{|x_N|= r_0}\left(w-A_0\frac{u}{r_N^2}\right), \limsup_{|x_N|\to \infty} \left(w-A_0\frac{u}{r_N^2}\right)\right\}
													\le
													0.
													\]
													Therefore, we arrive at
													\[
													0\le \sqrt{\rho}-u \le A_0\frac{u}{r_N^2}\le A_0\frac{\sqrt{\rho}}{r_N^2}
													\]
													for $|x_N|\ge r_0$, which proves
													\[
													\int_{\R^{3N-3}} \left| \frac{\psi(X_N)}{\sqrt{\rho(x_N)}}-\phi(X_{N-1})\right|^2\,dx_1\cdots dx_{N-1}
													=2\left(1-\frac{u}{\sqrt{\rho}}\right) \le Cr_N^{-2}.
													\]
													This is the desired result.
													\end{proof}

													\appendix
													\section{Asymptotic bounds for the ground state density}
													\label{App}
													In this appendix, we provide a proof of~\eqref{eq.HOM} and~\eqref{eq.threshold}.
													At the threshold $E_N=E_{N-1}$, we only consider the case $Z<N-1$, otherwise $E_N$ cannot be a ground state eigenvalue~\cite{Goto}.
													To prove these bounds, we rely on the following comparison theorem based on the maximum principle (see, e.g.~\cite[Thm.~2.1]{DHSV}).

													\begin{theorem}
														\label{thm.Acompa}
														Let $f\ge0$ and $g$ be continuous on $|x|\ge R$.
														Suppose that $f,g \to 0$ as $|x|\to\infty$ and
														\[
														(-\Delta +W_1)g\le 0\le (-\Delta+W_2)f
														\]
														in the distributional sense, and $f\ge g$ on $|x|=R$.
														If $W_1\ge W_2\ge0$, then we have $f\ge g$ on $|x|\ge R$.
														\end{theorem}

													To apply this theorem, we note that $u$ and $\sqrt{\rho}$ are continuous and vanish at infinity by the Sobolev embedding theorem.

													First, we prove the following upper bound.

													\begin{lemma}[Upper bounds for the one-particle densities]
														For any $x_N\in\R^3$ we have
														\begin{align}
														\sqrt{\rho(x_N)}&\le
														C(1+r_N)^{(Z-N+1)/\sqrt{2\epsilon_N}-1}\exp(-\sqrt{2\epsilon_N}r_N),\quad (E_N<E_{N-1}) \label{eqB.AHOMup}\\
														\sqrt{\rho(x_N)}&\le C(1+r_N)^{-3/4}\exp(-\kappa_N\sqrt{r_N}),\quad  (E_N=E_{N-1}) \label{eqB.thresholdup},
														\end{align}
															where $\epsilon_N=E_{N-1}-E_N$ and $\kappa_N:=\sqrt{8\left(N-1-Z\right)}$.
														\end{lemma}
														\begin{remark}
															The upper bounds also hold for physical ground state densities.
															\end{remark}

														Ahlrichs et al. gave~\eqref{eqB.AHOMup}~\cite[Thm.~3.1]{AHOM}, thus we only prove~\eqref{eqB.thresholdup}.
														\begin{proof}
															The following is an immediate consequence of Lemma~\ref{lem.Schroineq}; we also refer~\cite[Lemma~3.1]{AHOM}.
															\begin{lemma}
																There is a positive constant $d$ depending on $N,Z$ such that for sufficiently large $|x_N|$
																\[
																\left(
																-\frac{\Delta_N}{2}+\epsilon_N +\frac{N-1-Z}{r_N}-\frac{d}{r_N^2}
																\right)\sqrt{\rho(x_N)}\le0
																\]
																in the distributional sense.
																\end{lemma}
															Since we consider the threshold case, $N-1-Z>0$.
															For $a\in\R$, define
															\begin{equation}
																\label{eqB.defv}
															v_a(r_N):=Cr^{-3/4}\exp\left(-\kappa_N\sqrt{r_N}+\frac{a}{\sqrt{r}}\right).
															\end{equation}
															By a direct calculation, one has
															\begin{equation}
																\label{eqB.v}
																\begin{split}
																\left(
																-\frac{\Delta_N}{2}+\frac{N-1-Z}{r_N}-\frac{d}{r_N^2}
																\right)v_a(r_N)
																=
																\left(\frac{3/32 -d-\kappa_Na/4}{r^2} -\frac{a}{4r^{5/2}}-\frac{a^2}{8r^3} \right)v_a>0
																\end{split}
																\end{equation}
																for $a<4(3/32-d)/\kappa_N$ and large $r_N$.
																Choosing $C_R$ large enough so that for $|x_N|=R$
																\[
																C_Rv_a(r_N)\ge \max_{|x_N|=R}\sqrt{\rho(x_N)},
																\]
																which shows $C_Rv_a(r_N)\ge \sqrt{\rho(x_N)}$ from Theorem~\ref{thm.Acompa}.
																Since $\rho$ is bounded, we have~\eqref{eqB.thresholdup}.
															\end{proof}

															For the lower bound, we first show a preliminary lower bound.

															\begin{lemma}
																\label{lemB.ulowth}
																For sufficiently large $|x_N|$, we have
																\[
																u(x_N)\ge C_\lambda\exp(-\lambda|x_N|)
																\]
																for any $\lambda> \sqrt{2\epsilon_N}$.
																\end{lemma}

																\begin{proof}
																	The proof is essentially the same as that in~\cite{HO79}.
																	Let $\phi_R$ be the normalized Dirichlet ground state of $H_{N-1}$ confined in  $\Omega_R=\{X_{N-1}\colon |X_{N-1}|< R\}$: $H_{N-1}\phi_R=E_{N-1}^{(R)}\phi_R$.
																	We may assume $E_{N-1}^{(R)}-E_{N-1}<\lambda^2/2-\epsilon_N$ and define
																	\[
																	u_R(x_N):=\int_{\Omega_R}\phi_R(X_{N-1})\psi(X_N)\,dX_{N-1}.
																	\]
																	Applying Green's theorem leads to
																	\begin{align*}
																	0&=
																	\left(-\frac{\Delta_N}{2}-\frac{Z}{r_N}+E_{N-1}^{(R)} -E_N\right)u_R(x_N)
																	+\sum_{i=1}^{N-1}\int_{\Omega_R}\frac{\phi_R(X_{N-1}) \psi(X_N)}{|x_i-x_N|}\,dX_{N-1}\\
																	&\quad-\frac{1}{2}\int_{\partial\Omega_R} \left(\phi_R\frac{\partial \psi}{\partial n}-\psi \frac{\partial \phi_R}{\partial n}\right)\,dS_R.
																	\end{align*}
																	It is clear that
																	\[
																	\sum_{i=1}^{N-1}\int_{\Omega_R}\frac{\psi(X_N)\phi_R(X_{N-1})}{|x_i-x_N|}\,dX_{N-1}
																	\le
																	\frac{N-1}{r_N-R}u_R(x_N)
																	\]
																	for $|x_N|>R$.
																	Since $\phi_R=0$ on $\partial\Omega_R$ and $ {\partial \phi_R}/{\partial n}\le 0$, we have
																	\[
																	0\le
																	\left(-\frac{\Delta_N}{2}-\frac{Z}{r_N}+E_{N-1}^{(R)}-E_N+\frac{N-1}{r_N-R}\right)u_R(x_N).
																	\]
																	Since $E_{N-1}^{(R)}-E_N<\lambda^2/2$, we have $u_R\ge c_\lambda e^{-\lambda r_N}$ for $|x_N|$ large enough by Theorem~\ref{thm.Acompa}.
																	Combining with $\phi_R \le C_R\phi$, we have the conclusion.
																	\end{proof}

															\begin{lemma}[Lower bounds for one-particle densities]
																For any $x_N\in\R^3$, we have $\sqrt{\rho}\le Cu$ and
																\begin{align}
																	\sqrt{\rho(x_N)}&\ge u(x_N)\ge
																C(1+r_N)^{(Z-N+1)/\sqrt{2\epsilon_N}-1}\exp(-\sqrt{2\epsilon_N}r_N),\quad (E_N<E_{N-1}) \label{eqB.AHOMlow}\\
																	\sqrt{\rho(x_N)}&\ge  u(x_N)\ge C(1+r_N)^{-3/4}\exp(-\kappa_N\sqrt{r_N}),\quad (E_N=E_{N-1}) \label{eqB.thresholdlow},
																\end{align}
																\end{lemma}

																\begin{proof}
																	Let $\delta_N:=(E_{N-1}^{(1)}-E_{N-1})/2$.
																	From Lemma~\ref{lem.Schroineq}, we have
																	\begin{align*}
																		\left(-\frac{\Delta_N}{2}-\frac{Z-N+1}{r_N}+\epsilon_N+\delta_N-\frac{K}{r_N^2}\right)
																		\sqrt{\rho}
																		\le\delta_N u
																	\end{align*}
																	for sufficiently larger $K$.
																	Using~\eqref{eqA.ubound}, it follows that
																	\begin{align*}
																		\left(-\frac{\Delta_N}{2}-\frac{Z-N+1}{r_N}+\epsilon_N+\delta_N-\frac{K}{r_N^2}\right)(\sqrt{\rho}-2u)
																		\le \left( \frac{2K}{r_N^2} - \delta_N\right)u
																		\le 0
																		\end{align*}
																	for sufficiently large $|x_N|$.
																	Choosing $\lambda,\mu>0$ such that $\sqrt{2\epsilon_N}<\lambda< \mu<\sqrt{2(\epsilon_N+\delta_N)}$, the maximum principle yields $\sqrt{\rho}\le 2u+Ce^{-\mu r_N}$, since
																	\begin{equation}
																	\label{eqB.emubound}
																	\left(-\frac{\Delta_N}{2}-\frac{Z-N+1}{r_N}+\epsilon_N+\delta_N-\frac{K}{r_N^2}\right)
																	e^{-\mu r_N} \ge 0.
																	\end{equation}
																	By Lemma~\ref{lemB.ulowth}, we have $\sqrt{\rho}\le Cu$.
																	Hence~\eqref{eqA.ubound} becomes
																	\begin{equation}
																	\label{eqB.ubound}
																	\left(-\frac{\Delta_N}{2}-\frac{Z-N+1}{r_N} +\epsilon_N+\frac{D}{r_N^{2}} \right)u
																	\ge 0.
																	\end{equation}
																	Now we focus on the threshold case $\epsilon_N=0$.
																	Using the function $v_a$ in~\eqref{eqB.defv}, we have
																	\[
																	\left(-\frac{\Delta_N}{2}-\frac{Z-N+1}{r_N}+\frac{D}{r_N^{2}} \right)v_a
																	=
																	\frac{D+3/32-\kappa_N a/4}{r_N^2}v_a-\frac{a}{4r^{5/2}}v_a-\frac{a^2}{8r^3}v_a<0
																	\]
																	for large $a>0$ and $r_N$.
																	Theorem~\ref{thm.Acompa}  leads to $u\ge v_a$, which proves~\eqref{eqB.thresholdlow} since $u>0$ and $\sqrt{\rho}$ are continuous.

																	Next, we show~\eqref{eqB.AHOMlow}.
																	Let
																	\[
																	\beta=\frac{Z-N+1}{\sqrt{2\epsilon_N}}-1.
																	\]
																	Taking the function $g_a(r):=r^\beta \exp(-\sqrt{2\epsilon_N}r + a/r)$, we have
																	\[
																	\left(-\frac{\Delta_N}{2}-\frac{Z-N+1}{r_N}+\epsilon_N +\frac{D}{r_N^{2}} \right)g_a\le 0,
																	\]
																	for large $a>0$.
																	Then we conclude $u\ge cg_a \ge cr^\beta e^{-\sqrt{2\epsilon_N} r_N}$ for large $|x_N|$.
																	Since $u>0$ is continuous, we have the desired bound~\eqref{eqB.AHOMlow}.
																	\end{proof}

							\section*{Acknowledgments}
							The research is partially supported by JSPS KAKENHI Grant Number JP26K16996.


\begin{thebibliography}{99}

	\bibitem{Agmon}
	S.~Agmon,
	\emph{Lectures on exponential decay of solutions of second-order
		elliptic equations: Bounds on eigenfunctions of $N$-body
		Schr\"odinger operators},
	Mathematical Notes, vol.~29,
	Princeton University Press, Princeton, NJ, 1982.

	\bibitem{AHO}
	R.~Ahlrichs, M.~Hoffmann-Ostenhof, and T.~Hoffmann-Ostenhof,
	Bounds for the long-range behavior of electronic wavefunctions,
	\emph{J. Chem. Phys.} \textbf{68} (1978), 1402--1410.

	\bibitem{AHOM}
	R.~Ahlrichs, M.~Hoffmann-Ostenhof, T.~Hoffmann-Ostenhof, and
	J.~D.~Morgan~III,
	Bounds on the decay of electron densities with screening,
	\emph{Phys. Rev. A} \textbf{23} (1981), 2106--2117.

	\bibitem{AS}
	M.~Aizenman and B.~Simon,
	Brownian motion and Harnack inequality for Schr\"odinger operators,
	\emph{Comm. Pure Appl. Math.} \textbf{35} (1982), 209--273.

	\bibitem{BFLS}
	J.~Bellazzini, R.~L.~Frank, E.~H.~Lieb, and R.~Seiringer,
	Existence of ground states for negative ions at the binding threshold,
	\emph{Rev. Math. Phys.} \textbf{26} (2014), 1350021.

	\bibitem{BeSa}
	H.~A.~Bethe and E.~E.~Salpeter,
	\emph{Quantum mechanics of one- and two-electron atoms},
	Springer-Verlag, Berlin, Heidelberg, 1957.

	\bibitem{Briet}
	P.~H.~Briet,
	The relation between the $(N)$- and $(N-1)$-electron atomic
	ground states,
	\emph{J. Math. Phys.} \textbf{26} (1985), 2560--2564.

	\bibitem{Carmona}
	R.~Carmona,
	Regularity properties of Schr\"odinger and Dirichlet semigroups,
	\emph{J. Funct. Anal.} \textbf{33} (1979), 259--296.

	\bibitem{CS}
	R.~Carmona and B.~Simon,
	Pointwise bounds on eigenfunctions and wave packets in $N$-body
	quantum systems.~V. Lower bounds and path integrals,
	\emph{Comm. Math. Phys.} \textbf{80} (1981), 59--98.

	\bibitem{CHO}
	J.~M.~Combes, M.~Hoffmann-Ostenhof, and T.~Hoffmann-Ostenhof,
	Asymptotics of atomic ground states: The relation between the
	ground state of helium and the ground state of $\mathrm{He}^{+}$,
	\emph{J. Math. Phys.} \textbf{22} (1981), 1299--1305.

	\bibitem{CT}
	J.~M.~Combes and L.~Thomas,
	Asymptotic behaviour of eigenfunctions for multiparticle
	Schr\"odinger operators,
	\emph{Comm. Math. Phys.} \textbf{34} (1973), 251--270.

	\bibitem{DHSV}
	P.~Deift, W.~Hunziker, B.~Simon, and E.~Vock,
	Pointwise bounds on eigenfunctions and wave packets in $N$-body
	quantum systems.~IV,
	\emph{Comm. Math. Phys.} \textbf{64} (1978), 1--34.

	\bibitem{Goto}
	Y.~Goto,
	Absence of a ground state for bosonic Coulomb systems with
	critical charge,
	\emph{Rep. Math. Phys.} \textbf{81} (2018), 177--184.

	\bibitem{Gridnev}
	D.~K.~Gridnev,
	Bound states at threshold resulting from Coulomb repulsion,
	\emph{J. Math. Phys.} \textbf{53} (2012), 102108.

	\bibitem{KalfHinz}
	A.~M.~Hinz and H.~Kalf,
	Subsolution estimates and Harnack's inequality for
	Schr\"odinger operators,
	\emph{J. Reine Angew. Math.} \textbf{404} (1990), 118--134.

	\bibitem{HO79} Hoffmann-Ostenhof, T. Lower and upper bounds to the decay of the ground state one-electron density of helium-like systems.
	\emph{J. Phys. A: Math. Gen.} \textbf{12} (1979), 1181--1187.

	\bibitem{HH}
	M.~Hoffmann-Ostenhof and T.~Hoffmann-Ostenhof,
	``Schr\"odinger inequalities'' and asymptotic behavior of the
	electron density of atoms and molecules,
	\emph{Phys. Rev. A} \textbf{16} (1977), 1782--1785.

	\bibitem{HOS}
	M.~Hoffmann-Ostenhof, T.~Hoffmann-Ostenhof, and B.~Simon,
	A multiparticle Coulomb system with bound state at threshold,
	\emph{J. Phys. A: Math. Gen.} \textbf{16} (1983), 1125--1131.

	\bibitem{Kato}
	T.~Kato,
	On the eigenfunctions of many-particle systems in quantum mechanics,
	\emph{Comm. Pure Appl. Math.} \textbf{10} (1957), 151--177.

	\bibitem{KaDa}
	J.~Katriel and E.~R.~Davidson,
	Asymptotic behavior of atomic and molecular wave functions,
	\emph{Proc. Natl. Acad. Sci. USA} \textbf{77} (1980), 4403--4406.

	\bibitem{LiSe}
	E.~H.~Lieb and R.~Seiringer,
	\emph{The stability of matter in quantum mechanics},
	Cambridge University Press, Cambridge, 2010.

	\bibitem{LS}
	E.~H.~Lieb and B.~Simon,
	Pointwise bounds on eigenfunctions and wave packets in $N$-body
	quantum systems.~VI. Asymptotics in the two-cluster region,
	\emph{Adv. in Appl. Math.} \textbf{1} (1980), 324--343.

	\bibitem{Morgan}
	J.~D.~Morgan~III,
	The exponential decay of sub-continuum wavefunctions of
	two-electron atoms,
	\emph{J. Phys. A: Math. Gen.} \textbf{10} (1977), L91--L93.

	\bibitem{SimonS}
	B.~Simon,
	Schr\"odinger semigroups,
	\emph{Bull. Amer. Math. Soc. (N.S.)} \textbf{7} (1982),
	447--526.

	\bibitem{Simon}
	B.~Simon,
	\emph{Functional integration and quantum physics},
	2nd ed., AMS Chelsea Publishing, Providence, RI, 2005.

	\bibitem{Teschl}
	G.~Teschl,
	\emph{Mathematical methods in quantum mechanics: With applications
		to Schr\"odinger operators},
	2nd ed., Graduate Studies in Mathematics, vol.~157,
	American Mathematical Society, Providence, RI, 2014.

	\bibitem{Zhislin}
	G.~M.~Zhislin,
	A study of the spectrum of the Schr\"odinger operator for a system
	of several particles,
	\emph{Tr. Mosk. Mat. Obs.} \textbf{9} (1960), 81--120
	(in Russian).

\end{thebibliography}
	\end{document}